\documentclass{elsart}               % The use of LaTeX2e is preferred.

\usepackage{graphicx}          % Include this line if your 
\usepackage{natbib} 

\usepackage{amsmath,amssymb,color}

\newtheorem{theorem}{Theorem}

\newtheorem{lemma}[theorem]{Lemma}
\newtheorem{proposition}[theorem]{Proposition}
\newtheorem{definition}{Definition}
\newtheorem{remark}{Remark}
\newtheorem{example}{Example}
\newtheorem{assumption}{Assumption}

\DeclareMathOperator{\interior}{int}
\newcommand{\inte}[1]{\interior{\mbb{#1}}}

\newcommand{\mbb}[1]{\mathbb #1}

\newcommand{\mcl}[1]{\mathcal #1}

\newcommand{\Rnx}{\mbb{R}^{n_x}}

\newcommand{\Rnu}{\mbb{R}^{n_u}}

\newcommand{\R}{\mbb{R}}

\definecolor{darkblue}{rgb}{0.4, 0.0, 0.1}
\def\TF{\textcolor{black}}

\makeatletter
\DeclareRobustCommand{\qed}{%
  \ifmmode % if math mode, assume display: omit penalty etc.
  \else \leavevmode\unskip\penalty9999 \hbox{}\nobreak\hfill
  \fi
  \quad\hbox{\qedsymbol}}
\newcommand{\openbox}{\leavevmode
  \hbox to.77778em{%
  \hfil\vrule
  \vbox to.675em{\hrule width.6em\vfil\hrule}%
  \vrule\hfil}}
\newcommand{\qedsymbol}{\openbox}
\newenvironment{proof}[1][\proofname]{\par
  \normalfont
  \topsep6\p@\@plus6\p@ \trivlist
  \item[\hskip\labelsep\itshape
    #1.]\ignorespaces
}{%
  \qed\endtrivlist
}
\newcommand{\proofname}{Proof}
\makeatother
\date{2021}

\begin{document}

\begin{frontmatter}

\title{On Continuous-Time Infinite Horizon Optimal Control -- Dissipativity, Stability and Transversality\thanksref{footnoteinfo}} % Title, preferably not more 

\thanks[footnoteinfo]{This paper was not presented at any IFAC 
meeting. Corresponding author T.~Faulwasser. Tel. +49 231 755 2359. 
Fax +49 231 755 2694.}

\author[Paestum]{Timm Faulwasser}\ead{timm.faulwasser@ieee.org},    % Add the 
\author[Rome]{Christopher M. Kellett}\ead{chris.kellett@anu.edu.au }               % e-mail address 
%\author[Baiae]{Publius Maro Vergilius}\ead{vergilius@culture.ir}  % (ead) as shown

\address[Paestum]{Institute for Energy Systems, Energy Efficiency and Energy Economics, TU Dortmund University, Emil-Figge-Str. 70, 44227 Dortmund, Germany.}  % Please supply                                              
\address[Rome]{Research School of Electrical, Energy, and Materials Engineering at the Australian National University, Canberra, ACT, Australia.}             % full addresses
%\address[Baiae]{The White House, Baiae}        % here.

                             % keyword list or with the 
                                          % help of the Automatica 
                                          % keyword wizard

 \begin{abstract}                          % Abstract of not more than 200 words.
This paper analyses the interplay between dissipativity and stability properties in continuous-time infinite-horizon Optimal Control Problems (OCPs). We establish several relations between these properties, which culminate in a set of equivalence conditions. Moreover, we investigate convergence and stability of the infinite-horizon optimal adjoint trajectories. The workhorse for our investigations is a notion of strict dissipativity in OCPs, which has been coined in the context of economic model predictive control. With respect to the link between stability and dissipativity, the present paper can be seen as an extension of the seminal work on least squares optimal control by Willems from 1971.
Furthermore, we show that strict dissipativity provides a conclusive answer to the question of adjoint transversality conditions in infinite-horizon optimal control which has been raised by Halkin in 1974. Put differently, we establish conditions under which the adjoints converge to their optimal steady state value. We draw upon several examples to illustrate our findings. Moreover, we discuss the relation of our findings to \TF{results available in the literature}. 
\end{abstract}

\end{frontmatter}

\section{Introduction}
Arguably, the three most impactful concepts in systems and control in the 20th century have been the optimal control siblings---i.e. the Pontryagin Maximum Principle (PMP) \citep{Boltyanskii60a} and the Hamilton-Jacobi-Bellman  approach \citep{Bellman54a}---as well as the dissipativity notion for dynamic systems coined by \cite{Willems71a,Willems72a,Willems72b}.\footnote{Indeed the origins of optimal control theory can be traced back at least for another 300 years to the 17th century. For overviews of the historical development of optimal control theory see \citep{Sussmann97,Pesch13a}.} The intricate relations between the latter two concepts and stability properties of dynamic systems have been at the core of a number of seminal contributions in systems and control, see e.g. \citep{Kalman60a,Willems71a} or \citep{Moylan73a,Hill76a}.
Moreover, one can regard the manifold developments on Model Predictive Control (MPC) as an industrially successful attempt to overcome the difficulties of solving the Hamilton-Jacobi-Bellman Equation (HJBE) for closed-loop optimal controls by instead resorting to a receding horizon application of open-loop optimal controls \citep{	Rawlings17}---either obtained  via the PMP \citep{Kapernick14a} or via direct solution methods for optimal control \citep{Houska11a}. \vspace*{-3mm}%\footnote{Interestingly, the first conceptualizations of this idea can be traced back to \cite{Lee67} and \cite{Propoi63}.}

As a matter of fact, recent developments on MPC rely heavily on dissipativity notions of Optimal Control Problems (OCPs), see, e.g., \citep{Angeli12a,Mueller14a,Gruene13a,epfl:faulwasser15g}. Specifically, these developments are driven by the need to consider stage costs---i.e. Lagrange terms in the language of optimal control---beyond the established convex quadratic functions, which also goes under the label of \textit{economic MPC}, see \citep{kit:faulwasser18c} for a recent overview. A main driver for the development of these generalized MPC schemes have been so-called turnpike properties of OCPs, which are in essence similarity properties of OCPs \TF{parametric in the initial condition and the horizon length} \citep{Trelat15a,tudo:faulwasser21b}. While the term turnpike property was coined by \cite{Dorfman58}  and has received considerable attention in economics \citep{Carlson91}, it was not of significant interest in MPC until \citep{Rawlings09b,Gruene13a}. Indeed it can be shown that turnpike and dissipativity properties of finite-horizon OCPs are closely related and, under mild assumptions, equivalent \citep{Gruene16a,epfl:faulwasser15h}.\vspace*{-2mm}

In the present paper, we do not investigate MPC. Rather we are interested in analyzing the interplay between dissipativity of  infinite-horizon OCPs and the stability of the considered dynamics under optimal infinite-horizon controls. Put differently, we exploit dissipativity concepts to establish a relation between the PMP and stability properties of optimally controlled systems. We show that under mild assumptions asymptotic stability of the state  and control variables (i.e. primal variables) is equivalent to strict dissipativity of the underlying OCPs. Moreover, we also extend our analysis to the (dual) adjoint/co-state variables of the OCP. We establish a set of conditions showing equivalence of dissipativity of an OCP and the stability of (primal and dual)  optimal infinite-horizon trajectories. Indeed, one may regard the present paper as an extension of the classic linear-quadratic analysis by \cite{Willems71a} to non-linear systems with generic cost functions and subject to constraints. \vspace*{-2mm}

Finally---in some sense as a by-product of our analysis and, in another sense, one of the key findings of this paper---we show that strict dissipativity properties of an OCP allow conclusively answering the question for adjoint transversality conditions of infinite-horizon OCPs, an open problem since the seminal paper of  \cite{Halkin74a}.
Specifically, we show that whenever the considered OCP is strictly dissipative then the optimal adjoint will converge to its steady-state value, which can be different from $0$ and corresponds to the optimal Lagrange multiplier of a corresponding steady-state optimization problem. Since Halkin's counterexamples, there have been different approaches to infinite-horizon transversality conditions. The findings of \cite{Pickenhain06,Pickenhain10a} rely on weighted Banach spaces (i.e. discounted objectives), \cite{Weber06a} considers exponentially discounted objectives to derive transversality bounds, while \cite{Cartigny03a} require structural properties to enforce boundedness of the infinite-horizon objective. 
Our approach structurally differs from these works as we rely on  strict dissipativity of optimal solutions, which enables us to show strong optimality (i.e. finiteness of the optimal value function) without discounting by a simple shift of the stage cost, which in turn alters neither primal nor dual optimal solutions. \vspace*{-2mm}

The remainder is structured as follows:
 Section \ref{sec:problem} presents the problem at  hand and recalls optimality conditions and dissipativity inequalities, while Section \ref{sec:results} presents the main results in the following order: primal attractivity and stability, converse dissipativity results, adjoint stability and transversality conditions, and equivalence conditions. Owing to the widespread investigations on and applications of  dissipativity, stability, and optimal control in the literature, we deviate from the customary contextualization of our results in the introduction. Instead Section \ref{sec:discussion} discusses  our results and puts them in context to related topics, such as e.g. viscosity solutions of the HJBE. Finally, the paper ends with brief conclusions in Section \ref{sec:conclusion}.

\section{Problem statement} \label{sec:problem}
We investigate time-invariant (finite or infinite-horizon)  OCPs in Lagrange form given by
\begin{subequations} \label{eq:OCP}
\begin{align}
V_{T}(0,x_0) \doteq &\inf_{u(\cdot) \in \mathcal{L}^\infty([0,T], \Rnu) } \int_0^{T}
\hspace*{-2mm}\ell(x(t), u(t)) \mathrm{d}t  \label{eq:OCP_obj}\\
\text{subject to }& \nonumber \\
\dfrac{\mathrm{d} x}{\mathrm{d}t} &= f(x(t),u(t)), ~ x(0) = x_0 %\in \mathbb{X}_0  
\label{eq:OCP_sys}\\
0 &\geq g_i(x(t), u(t)), ~ i = 1 \dots n_g \label{eq:OCP_con}
\end{align}
\end{subequations}
where $T \in \R^+\cup\infty$.
The dynamics  $f:\mathbb{R}^{n_x}\times \mathbb{R}^{n_u} \to \mathbb{R}^{n_x}$, the stage cost $\ell:\mathbb{R}^{n_x}\times \mathbb{R}^{n_u} \to \mathbb{R}$, and the mixed input-path constraints $g_i:\mathbb{R}^{n_x}\times \mathbb{R}^{n_u} \to \mathbb{R}, i = 1 \dots n_g$ are at least twice continuously differentiable.
Occasionally, we denote the constraint set defined by \eqref{eq:OCP_con} as
\begin{equation}\label{eq:Z}
\mbb{Z} \doteq \left\{(x,u) \in \mbb{R}^{n_x+n_u}\,|\, g_i(x,u)\leq 0, i=1\dots n_g\right\}.
\end{equation}
The projection of $\mbb{Z}$ onto $\Rnx$ is denoted by $\mbb{X} \doteq \Pi_x(\mbb{Z})$, and the projection onto $\Rnu$ is written as $\mbb{U} \doteq \Pi_u(\mbb{Z})$.

We assume that for admissible inputs $u(\cdot) \in$ $\mathcal{L}^\infty([0,T],$ $\Rnu)$ satisfying the constraints, the dynamics  \eqref{eq:OCP_sys} have a unique absolutely continuous solution. 
 Moreover, we suppose that for all initial conditions of interest, i.e. $x_0 \in \mathbb{X}_0 \subseteq \mathbb{R}^{n_x}$, an optimal solution exists. %, i.e. the optimal state response is absolutely continuous. 
 Note, at this point, we still need to comment on the specific optimality concept (strong or overtaking optimality) employed, as in the infinite-horizon case the performance functional \eqref{eq:OCP_obj} might be unbounded, cf. \citep{Carlson91} and Lemma \ref{lem:strgOpt} below.
We denote optimal pairs as 
\[z^\star(\cdot, x_0) \doteq (u^\star(\cdot, x_0),\, x^\star(\cdot, x_0)),
\]
 where the argument $x_0$ is used to denote the specific initial condition. Whenever necessary, we use $x^\star(\cdot, x_0, u^\star(\cdot, x_0))$ to highlight the considered input trajectory. \vspace*{-2mm}
 
As a shorthand for the infinite-horizon variant of \eqref{eq:OCP} we use OCP$_\infty(x_0)$, which highlights the considered horizon length and the initial condition $x_0$. Similarly, OCP$_T(x_0)$ refers to the finite horizon variant of \eqref{eq:OCP}. Any variable related to  OCP$_T(x_0)$ will be indicated by subscript $(\cdot)_T$ whenever necessary.

Subsequently, we investigate the stability of the dynamics \eqref{eq:OCP_sys} under the open-loop infinite-horizon optimal control
$u^\star: \mathbb{R}_0^+\times \mathbb{X}_0 \to \mathbb{R}^{n_u}$, i.e. 
\begin{equation} \label{eq:sys}
\dot x = f(x, u^\star(t, x_0)) \doteq f^\star(t, x), \quad x_0 \in \mathbb{X}_0.
\tag{$\Sigma$}
\end{equation} 
\TF{Temporarily assume} that the optimal control $u^\star(\cdot, x_0)$ of OCP$_\infty(x_0)$ is unique almost everywhere, then we know from Bellman's principle of optimality that the truncation of $u^\star(\cdot, x_0)$ to the time horizon $[\delta, \infty)$ is optimal for OCP$_\infty(x_\delta)$ with $x_\delta \doteq x^\star(\delta, x_0, u^\star(\cdot, x_0))$. 
Hence, 
the dynamics \eqref{eq:sys} have the following (semigroup) property
\begin{equation} \label{eq:semigroup_prop}
x(t+\delta, x_0, u^\star(\cdot, x_0)) = x(t, x_\delta, u^\star(\cdot, x_\delta)).\vspace*{-1.5mm}
\end{equation}

\begin{remark}[Non-unique optimal solutions]~\\
Evidently, it is restrictive to assume that OCP$_\infty(x_0)$ admits a.e. unique optimal solutions. If, for some $x_0 \in \mathbb{X}_0$, there exist multiple optimal solutions, we choose one of them at $t=0$ and stick to the corresponding optimal input on $[0, \infty)$. This way, the system \eqref{eq:sys} is uniquely defined. %\vspace*{-1.5mm} 
\TF{Moreover, our subsequent results hold as long one as one does not switch from one optimal input to another. Thus non-uniqueness of optimal solutions does not pose further issues.}
\end{remark}

\subsection{Necessary conditions of optimality}
To handle the mixed input-state (path) constraints \eqref{eq:OCP_con}, we consider a direct-adjoining approach via the 
Hamiltonian $H:\mathbb{R}\times \mathbb{R}^{n_x}\times \mathbb{R}^{n_g} \times \mathbb{R}^{n_x} \times \mathbb{R}^{n_u} \to \mathbb{R}$
\begin{equation}
\label{eq:H}
H(\lambda_0, \lambda, \mu, x, u) \doteq \lambda_0 \ell(x,u) + %\lambda^\top f(x,u)+\mu^\top g(x,u).
\begin{bmatrix}
\lambda \\ \mu
\end{bmatrix}^\top
\begin{bmatrix}
f(x,u) \\ g(x,u)
\end{bmatrix}
\end{equation}
The gradients of $H$ with respect to $x,u,\lambda$ are written as $H_x, H_u, H_\lambda$, respectively.\footnote{Occasionally, we will also use $\nabla W$ to denote the gradient of functions $W:\Rnx\to \R$.}

We exclude abnormal problems and hence we normalize $\lambda^\star_0 =1$. Applying the  PMP, first-order necessary conditions of optimality are given by
\begin{subequations} \label{eq:NCO}
\begin{align}
\dot x^\star &=\phantom{-} H_\lambda, \quad x^\star(0) = x_0, \label{eq:NCO_x}\\
\dot \lambda^\star &= - H_x,  \label{eq:NCO_adjoint}\\
0{\phantom{^\star}} & = \phantom{-}H_u, \label{eq:NCO_u}\\
H(\lambda_0^\star, \lambda^\star, \mu^\star, x^\star, u^\star)  &= \min_{u\in \mbb{U}} H(\lambda_0^\star, \lambda^\star, \mu^\star, x^\star, u). 
%\\0 & \leq \mu.
\end{align}
The conditions above are augmented by 
\begin{align}
%\dfrac{\mathrm{d} H}{\mathrm{d} t} &= 0, \quad
%H_T(\lambda_0^\star, \lambda^\star(t), \mu^\star(t), x^\star(t), u^\star(t)) &= const. \\
\hspace*{-6mm}H(\lambda_0^\star, \lambda^\star, \mu^\star, x^\star, u^\star) &= \begin{cases} 0\phantom{onst}  ~\text{ if } T = \infty \\ 
 const~ \text{ if } T < \infty
\end{cases} \label{eq:HamiltonZero}\\
\lambda^\star(T) & = 0 \quad\,\phantom{onst}\text{ if } T < \infty \label{eq:transversalityAdjoint} \\
(\mu^\star)^\top g(x^\star, u^\star) = 0,~ \mu^\star(t)  &\geq 0 \quad ~~ \phantom{onst}t\in [0, T),
\end{align}
\end{subequations}
and the usual non-triviality requirement that the adjoints may not vanish simultaneously.
Moreover, whenever $V_T \in \mcl{C}^1$ and no state constraints are active, we have that
\begin{equation} \label{eq:AdjointNablaV}
\nabla V_{T}(t,x^\star_T(t, x_0)) = \lambda^\star_T(t, x_0).
\end{equation}
Here the gradient of the optimal value function for the truncated horizon $[t, T]$ evaluated at $x^\star_T(t, x_0)$ is denoted as $\nabla V_{T}(t, x^\star_T(t, x_0))$. The gradient  equals the adjoint $ \lambda^\star_T(t, x_0)$ at time $t$ \citep{Liberzon12}. Whenever no confusion can arise the time argument of $V_T(t,x)$ is dropped.
\begin{remark}[Constraint qualifications in OCPs]
We remark that along optimal solutions one requires the mixed input-state constraints to be regular (i.e. linearly independent and of full rank with respect to $u$)  and the existence of multiplier trajectories $\mu(\cdot)$. Hence our standing assumption is that \eqref{eq:NCO} hold for  $u^\star(\cdot) \in \mathcal{L}^\infty([0,T], \Rnu)$, 
$x^\star(\cdot)$ absolutely continuous, and  $\lambda^\star(\cdot), \mu^\star(\cdot)$ are piecewise absolutely continuous. Notice that whenever no state constraints are present, i.e. $g(x,u)$ does not depend on $x$, this is not a severe restriction, see \cite{Hartl95}. Indeed we could drop \eqref{eq:NCO_u} and the multiplier $\mu$ and work with the usual Hamiltonian instead.  Alternatively, one could consider optimality conditions formulated in terms of bounded variation. Here, we restrict the discussion to the more easily accessible case of \eqref{eq:NCO}, which allows to highlight structural properties. For an overview and discussion of these and further necessary conditions we refer to \cite{Hartl95}.
\end{remark}

Moreover, it is worth noting that the steady-state variant of the optimality system \eqref{eq:NCO}
\begin{subequations} \label{eq:KKT}
\begin{equation}
0 = H_\lambda, \quad 0 = H_x, \quad 0 = H_u %, \quad \mu \geq 0
\end{equation}
combined with 
\begin{equation}
\mu \geq 0 \quad \text{ and } \quad  \mu^\top g = 0
\end{equation}
\end{subequations}
specifies the KKT conditions of the following steady-state optimization problem
\begin{subequations} \label{eq:SOP}
\begin{align}
\min_{(x,u) \in \mathbb{R}^{n_x + n_u}} ~ &\ell(x,u) \\
\text{ subject to } \nonumber \\
0 &=f(x,u), \\
0 &\geq g_i(x,u), ~  i = 1 \dots n_g. 
\end{align}
\end{subequations}
Optimal variables at steady state are denoted by $\bar\cdot$. Similarly to before, we use the shorthand $\bar z = (\bar x, \bar u)^\top$ to denote the optimal steady state.\vspace{2mm}

Observe \eqref{eq:NCO} does not specify a boundary condition for the adjoints $\lambda$ in the optimality conditions for $T = \infty$.  Indeed as the next classical example shows, the usual (finite horizon) transversality condition \eqref{eq:transversalityAdjoint} does not necessarily hold asymptotically in the infinite horizon case.

\begin{example}[ The example of \cite{Halkin74a}]\label{ex:halkin}
Consider OCP$_T(x_0)$ with 
$$-\ell(x,u) = (1-x)u = f(x,u), \qquad \dim x = \dim u = 1,$$ 
input constraint $\mbb{U} = [0,1]$, and horizon $T = \infty$.
It can be shown that the optimal solution is $u^\star(t) \equiv 1$ \cite[Chap. 2.4]{Carlson91}. This implies that the adjoint reads
$\lambda^\star(t) = (\lambda^\star(0) +\lambda^\star_0)e^{t}-\lambda^\star_0$. 
Upon normalization of $ -\lambda^\star_0 \doteq \lambda^\star(0)$ we obtain $\lambda^\star(t) \equiv \lambda^\star_0\not = 0$, which clearly differs from $\displaystyle \lim_{t\to\infty} \lambda^\star(t) = 0$. 
% \vspace*{2mm}
\end{example} 

The next example is taken from \cite{Carlson91,Cliff73}, wherein a finite-horizon variant is considered. It also shows the difficulties surrounding the transversality condition of the adjoints in the infinite-horizon case, and  illustrates the tight relation between OCP$_T(x_0)$ and the corresponding steady-state problem \eqref{eq:SOP}. 
\begin{example}[Optimal fish harvest]\label{ex:fish}
Consider the dynamics and stage cost
\begin{align*}
f(x,u) &= \,\phantom{-}x(x_s -x -u), \quad x_s > 0 \\
\ell(x,u) & = -ax -bu +cxu, \quad a,b,c > 0
\end{align*}  
with data $\mbb{U} = [0,\hat u]$ and $\mbb{X} = [\varepsilon, \hat x]$ and $T = \infty$.
Consider the  steady state  $
 \bar x = \frac{c x_s + b -a}{2c} 
 $
 with data $a,b,c$ such that $\bar x \in [\varepsilon, x_s]$, 
 The optimal closed-loop control is given by 
\[
u^\star(x) = \begin{cases} 0 & \text{ if } x < \bar x \\
\hat u & \text{ if } x > \bar x \\
 x_s - \bar x & \text{ if } x = \bar x, 
\end{cases}
\]
where we assume that $a,b,c$ are such that $\bar u \in (0, \hat u)$. 
Moreover, it can be shown that 
$\lim_{t \to \infty} x^\star(t) = \bar x$, that $\lim_{t \to \infty} u^\star(t) = \bar u$,
and that 
\[
\lim_{t \to \infty} \lambda^\star(t) = \bar \lambda = c - \dfrac{b}{\bar x}.
\]
 Observe that, for suitable values of the parameters $a,b,c$ and $x_s$ $(\bar  x, \bar u, \bar \lambda, 0)$ constitutes a KKT point and global minimizer of 
 %the steady-state problem 
\begin{align*}
\min_{x, u}~&-ax -bx + cxu \\
\text{subject to }& \\
0 &= x(x_s-x-u), \quad 
u \in [0, \hat u] \text{ and } x \in [\varepsilon, \hat x].\vspace*{-2mm}
\end{align*} 
Details of the derivation for the finite-horizon case can be found in \cite[Chap. 3.3]{Carlson91}.
%We will revisit this example below.  
\end{example}

\subsection{Dissipativity of OCPs}
We are interested in analyzing OCP$_\infty(x_0)$ 	under the following dissipativity assumption:
\begin{definition}[Dissipativity of OCP$_T(x_0)$]\label{def:diss}~\\
OCP$_T(x_0)$  \eqref{eq:OCP} is said to be \textit{dissipative with respect to $ \bar z = (\bar x, \bar u)$} if 
there exists a non-negative storage function $S:\mathbb{X} \to \mathbb{R}_0^+$ such that for all $x_0 \in \mathbb{X}_0$, all $T\geq 0$, and along all optimal pairs $z^\star(\cdot, x_0)$ of \eqref{eq:OCP_con} for all $t_1 \in [0,T]$ 
%\vspace*{-1mm}
%\begin{subequations} \label{eq:DI}
\begin{equation}\label{eq:diss}
%\vspace*{-2mm} 
S(x^\star(t_1)) - S(x_0) \leq \int_{0}^{t_1} \ell(z^\star(t)) -  \ell(\bar z)\,\mathrm{d}t \tag{DI}
 \end{equation}
holds, where  $x^\star(t_1) = x^\star(t_1, x_0, u^\star(\cdot))$. \\
If, in addition there exists $\alpha_\ell \in \mathcal{K}_\infty$ such that
%\vspace*{-2mm}
\begin{multline} \label{eq:strDI}
S(x^\star(t_1)) - S(x_0) \leq \int_{0}^{t_1} -\alpha_\ell\left(\left\|z^\star(t)-\bar z\right\|\right) 
+ \ell(z^\star(t))-  \ell(\bar z)\,\mathrm{d}t, \tag{sDI}
\end{multline}
%\end{subequations}
then OCP$_T(x_0)$ from  \eqref{eq:OCP} is said to be \textit{strictly dissipative with respect to $ \bar z = (\bar x, \bar u)$}. 
 \end{definition}
 It is easy to see that whenever strict dissipativity holds the steady state pair $(\bar x, \bar u)$ in \eqref{eq:strDI} is the unique global minimizer of \eqref{eq:SOP}. Henceforth, without loss of generality, we set $ \ell(\bar z) = \ell(\bar x, \bar u) = 0$. Note that swapping $\ell(x,u)$ with $\ell(x,u) -\ell(\bar z)$  neither affects the optimality of primal lifts $z^\star(\cdot,x_0)$ nor that of the optimal duals $\lambda^\star(\cdot,x_0)$ and $\mu^\star(\cdot,x_0)$. However, as we will see in Lemma \ref{lem:strgOpt}, this trick affects boundedness of $V_\infty$. \vspace{-2mm}
 
A classical characterization of dissipativity is given by the available storage \citep{Willems72a}. Let $\mathcal{U}^\star_T(x_0)$ denote the set of all optimal input trajectories of OCP$_T(x_0)$ for a given horizon length $T$ and initial condition $x_0$. In case of strict dissipativity of OCP$_T(x_0)$ and assuming w.l.o.g. $\ell(\bar z) =0$, the available storage is given by 
\begin{subequations} \label{eq:strAvailStorage}
\begin{align}
S^a_{\alpha,\ell}(x_0) \doteq \sup_{T} &\,\int_0^{T}%\hspace*{-3mm}
	\alpha_\ell\left(\left\|z(t)-\bar z\right\|\right) -\ell(z(t)) \mathrm{d}t   \label{eq:Sa_obj}	\\
\text{subject to }& \nonumber \\	
\dfrac{\mathrm{d} x}{\mathrm{d}t} &= f(x(t),u(t)), ~ x(0) = x_0 \\
u(\cdot) &\in \mathcal{U}^\star_T(x_0),
\end{align}
\end{subequations}
where the control signals are restricted to be optimal in OCP$_T(x_0)$.
Strict dissipativity of OCP$_T(x_0)$  in the sense of  \eqref{eq:strDI} is equivalent to $S^a_{\alpha,\ell}(x_0) <\infty$ for all $x_0 \in \mathbb{X}_0$ \citep{Willems72a}. 
The available storage for non-strict dissipativity based on \eqref{eq:diss} is given by
\begin{subequations} \label{eq:availStorage}
\begin{align}
S^a_{\ell}(x_0) \doteq \sup_{T} &\,\int_0^{T}%\hspace*{-3mm}
 -\ell(z(t)) \mathrm{d}t\\%  \label{eq:Sa_obj}\\
\text{subject to }& \nonumber \\
\dfrac{\mathrm{d} x}{\mathrm{d}t} &= f(x(t),u(t)), ~ x(0) = x_0 \\
u(\cdot) &\in \mathcal{U}^\star_T(x_0),
\end{align}
\end{subequations}
Note that strict dissipativity implies dissipativity. 

 \begin{remark}[Dissipativity notions for OCPs]~\\
 We remark that there exist slightly differing dissipativity notions for OCPs. Some works consider dissipation inequalities to hold for all $(x,u) \in \mbb{Z}$ \citep{Mueller14a}. Other works \citep{epfl:faulwasser15g,epfl:faulwasser15h,kit:faulwasser18e_2} require dissipativity only along optimal solutions, which is a slightly weaker requirement. Moreover, here we consider strictness in $x$ and $u$, while occasionally strictness in $x$ is used, see \citep{Angeli12a,Mueller14a}.\vspace*{-1mm}
 \end{remark}

\section{Results} \label{sec:results}
We present our result first for the primal variables $x, u$ and then we shift to the dual/adjoint variables $\lambda, \mu$.
\subsection{Primal attractivity and stability}~\vspace*{-12mm}
\begin{assumption}[Exponential cost bound] \label{ass:reach}~\\
For all $x_0 \in \mathbb{X}_0$, there exists an admissible infinite-horizon control $\tilde u:[0, \infty) \to \mathbb{R}^{n_u}$ and constants $C(x_0) >0, \rho >0$
such that the feasible suboptimal pair $\tilde z(\cdot, x_0) = (\tilde x(\cdot, x_0, \tilde u(\cdot), ~ \tilde u(\cdot))$ satisfies
\[
\ell(\tilde z(t, x_0)) \leq C(x_0) e^{-\rho t} %. \vspace*{-3mm}
\]
and \TF{$C(x_0)$ is finite on any compact subset of $\mbb{X}_0$}.
\end{assumption}
%This assumption allows to bound $V_\infty$ from above by a finite constant. 
\begin{lemma}[\eqref{eq:diss} $\Rightarrow$ $V_\infty(x) >- \infty$] \label{lem:strgOpt}
For all $ x_0 \in \mathbb{X}_0$, let OCP$_T(x_0)$  be dissipative with respect to $\bar z = (\bar x, \bar u)^\top$, $\ell(\bar z) = 0$, and let Assumption \ref{ass:reach} hold. Then, there exists a constant  $\overline{v}>0$ such that \TF{on any compact subset of $\mathbb{X}_0$
\[
-\infty <  V_T(x_0) \leq \overline v < \infty,\qquad \forall T\in \mbb{R}^+\cup \infty.  \vspace*{-3mm}
\]
holds  and $\displaystyle -\infty <  \inf_{T\in \mbb{R}^+} V_T(x_0)$ holds on  $\mathbb{X}_0$.}
\end{lemma}
\begin{proof}
The dissipativity characterization via the available storage \eqref{eq:availStorage} implies
\[
\infty > S^a_\ell(x) = \sup_T -\int_0^T \ell(z^\star_T(t,x_0))\mathrm{d}t  \\= \sup_T -V_T(x_0)
\]
\TF{Hence dissipativity gives $-\infty <  \inf_{T\in \mbb{R}^+} V_T(x_0)$ on  $\mathbb{X}_0$ and 
\[S^a_\ell(x) \geq  -\int_0^\infty \ell(z_\infty^\star(t,x_0))\mathrm{d}t = -V_\infty(x_0)\]} 
\TF{ Moreover, Assumption \ref{ass:reach} implies for all $T\in\mbb{R}^+\cup\infty$
$
\int_0^T \ell(z_T^\star(t, x_0)) \mathrm{d}t \leq \int_0^T \ell(\tilde z(t, x_0))\mathrm{d}t  \leq \tilde J(x_0) <\infty$.
Hence  we have $-\infty <  V_T(x_0) \leq \tilde J$ for all $T\in \mbb{R}^+\cup \infty$ on any compact subset of $\mathbb{X}_0$.}
\end{proof}
The insight obtained from the above lemma is that dissipativity of OCP$_\infty(x_0)$ 
implies strong optimality. Hence, we do not need to
resort to more general concepts such as \textit{weak or strong  overtaking} optimality \citep{Carlson91}. Also observe that the above proof does not hinge on strictness of dissipativity. 

\begin{theorem}[\eqref{eq:strDI} $\Rightarrow$ primal attractivity]\label{thm:attractivity}~\\
For all $ x_0 \in \mathbb{X}_0$, let OCP$_T(x_0)$  be strictly dissipative with respect to $\bar z = (\bar x, \bar u)^\top$ and let Assumption \ref{ass:reach} hold. 
Then, for all $ x_0 \in \mathbb{X}_0$, the solutions of \eqref{eq:sys} satisfy 
\[
\lim_{t\to\infty} x(t, x_0, u^\star(\cdot, x_0)) = \bar x. 
\]
Furthermore, if there exists an optimal infinite-horizon input $u^\star(\cdot, x_0)$ absolutely continuous on $[0, \infty)$, then
\[
\lim_{t\to\infty} u^\star(t, x_0) = \bar u.  \vspace{-3mm}
\]
\end{theorem}
\begin{proof}
For the sake of contradiction, assume that---despite  OCP$_T(x_0)$ being strictly dissipative---for some infinite-horizon optimal pair $\hat z(\cdot, x_0)$ generated by $\hat u(\cdot) \in \mathcal{U}^\star_\infty(x_0)$ we have
\[
\lim_{T\to \infty} \int_0^T\alpha_\ell\left(\left\|\hat z(t, x_0)-\bar z\right\|\right) \mathrm{d}t = \infty.
\]
Since $\alpha_\ell \in \mcl{K}_\infty$, there exists a subset $\mathcal{T} \subseteq [0,\infty)$  such that $$\alpha_\ell\left(\left\|\hat z(t, x_0)-\bar z\right\|\right) > 0, \quad \forall t\in\mathcal{T}$$ and $\alpha_\ell\left(\left\|(\hat z(t, x_0)-\bar z\right\|\right) = 0$ for all $t\in [0,\infty)\setminus\mathcal{T}$.

Hence, along $\hat z(\cdot, x_0)$  the functional characterizing $S^a_{\alpha, \ell}$ in \eqref{eq:Sa_obj} can be written as
\begin{multline*}
\int_{[0,\infty)\setminus\mathcal{T}} \alpha_\ell(\cdot) - \ell(\hat z(t, x_0))  \mathrm{d}t + 
\int_{\mathcal{T}} \alpha_\ell(\cdot) - \ell(\hat z(t, x_0))  \mathrm{d}t =\\
\int_0^\infty - \ell(\hat z(t, x_0))  \mathrm{d}t + 
\int_{\mathcal{T}} \alpha_\ell(\cdot)  \mathrm{d}t .
\end{multline*}
Observe that the first term in the last line corresponds to $-V_\infty(x_0)$. Thus we obtain
\[
\underbrace{-V_\infty(x_0)}_{>-\infty \text{ (Lem. \ref{lem:strgOpt})} } + \underbrace{\int_{\mathcal{T}} \alpha_\ell(\|\hat z(t, x_0)-\bar z\|)  \mathrm{d}t}_{=\infty} =\infty.
\]
This, however, means that along $\hat z(\cdot, x_0)$ the functional \eqref{eq:Sa_obj} 
equates to $\infty$, which in turn contradicts
 $S^a_{\alpha, \ell}<\infty$ and thus it also contradicts strict dissipativity. 
Hence, we have for all $u(\cdot) \in \mathcal{U}^\star_T(x_0)$ and all $x_0 \in \mathbb{X}_0$ that
\begin{equation}    \label{eq:alpha_funbnd}
\lim_{T\to \infty} \int_0^T\alpha_\ell\left(\left\|(z(t, x_0, u(\cdot)))-\bar z\right\|\right) \mathrm{d}t < \infty.
\end{equation}
Applying Barbalat's Lemma \cite[Lem. 4]{Michalska94} directly gives 
$\lim_{t\to\infty} x(t, x_0, u^\star(\cdot, x_0)) = \bar x.$
The second assertion follows also via Barbalat's Lemma from  absolute continuity of $u^\star(\cdot, x_0)$. \vspace*{-2mm}%fact that \eqref{eq:alpha_funbnd} implies
\end{proof}
Note that in the above proof the convergence of the optimal state trajectory requires strictness in \eqref{eq:strDI} only with respect to $x$. Likewise convergence of the control $u$ requires strictness in \eqref{eq:strDI} with respect to $u$.
\begin{remark}[Dissipativity and reachability]~\\
Theorem \ref{thm:attractivity}  highlights the tight interplay between dissipativity and reachability properties. Considering the foundations of dissipativity laid by \cite{Willems72a,Willems71a}---and in particular the definition of available storage and required supply---, Theorem \ref{thm:attractivity} is no surprise. In essence, the crucial strictness of \eqref{eq:strDI} expressed by $\alpha_\ell$ induces an implicit reachability property. 
We will see later in Proposition \ref{prop:expReach} that, under suitable regularity assumptions, dissipativity implies exponential reachability.  %\vspace*{1.5mm}
\end{remark}

\begin{lemma}[$V_\infty(\bar x) = 0$]\label{lem:Vinf0}
Let OCP$_T(x_0)$ be strictly dissipative at $\bar z = (\bar x, \bar u)\in\mbb{Z} $, let Assumption \ref{ass:reach} hold, and let $\ell(\bar z) = 0$, then $V_\infty(\bar x) = 0$. Moreover, the infinite-horizon state response satisfies $x^\star(t,\bar x) = \bar x$ on $[0,\infty)$ and $u^\star(t,\bar x) = \bar u$ holds a.e. on $[0,\infty)$.
\end{lemma}
\begin{proof}
First note that $(\bar x, \bar u, \bar \lambda, \bar \mu)$ (setting $\lambda_0 =1$) constitutes an equilibrium of \eqref{eq:NCO_x}-\eqref{eq:NCO_u}. Hence at  $x_0 = \bar x$ this equilibrium is an infinite-horizon admissible solution satisfying the necessary conditions \eqref{eq:NCO}. Since $\ell(\bar x, \bar u) =0$, we arrive at $V_\infty(\bar x) \leq  0$.
From Theorem \ref{thm:attractivity} we have that 
\[
x_\infty \doteq \lim_{t\to\infty} x(t, \bar x, u^\star(\cdot)) = \bar x. 
\]
For optimal solutions starting at $x_0 = \bar x$, the strict dissipation inequality \eqref{eq:strDI} can be written as
%\begin{multline*}
\[
S(x_\infty) - S(\bar x) \leq 
\int_0^\infty -\alpha_\ell(\left\|z^\star(t, \bar x)-\bar z\right\|)  \mathrm{d}t + V_\infty(\bar x).
\]
%\end{multline*}
Since $x_\infty = \bar x$ we obtain
\begin{equation} \label{eq:VinfLowBnd}
\int_0^\infty \alpha_\ell(\left\|z^\star(t, \bar x)-\bar z\right\|)  \mathrm{d}t \leq V_\infty(\bar x)\leq 0.
\end{equation}
Due to $\alpha_\ell \in \mcl{K}_\infty$, the integral is non-negative and thus  we arrive at
$0\leq V_\infty(\bar x) \leq 0$. 
Finally, for the sake of contradiction, suppose that $z_\infty^\star(t,x_0) \not\equiv \bar z = (\bar x, \bar u)$ on some set $\mcl{T} \subseteq [0,\infty)$ with Lebesgue measure $\nu[\mcl{T}] > 0$.  This would imply
\[
\int_\mcl{T} \alpha_\ell(\left\|z^\star(t, \bar x)-\bar z\right\|)  \mathrm{d}t > 0,
\]
which contradicts \eqref{eq:VinfLowBnd}.
\end{proof}
We consider $W:\Rnx\to\R$  
\begin{equation}\label{eq:W}
W(x) = V_\infty(x) + S(x) - S(\bar x)
\end{equation}
as a candidate Lyapunov function for \eqref{eq:sys}. This choice is motivated by results on stability analysis for economic MPC \citep{Gruene17a,epfl:faulwasser15g} and by fact that the semigroup property \eqref{eq:semigroup_prop} suggests a time-invariant Lyapunov function.
\begin{assumption} \label{ass:contW}
There exists $\alpha_W \in \mcl{K}$ and an open neighborhood of $\bar x$  such that
\[
\underline{\alpha}_W(\|x - \bar x\|)  \leq V_\infty(x) + S(x) -  S(\bar x)
\]
holds locally. 
\end{assumption}
Note the above assumption essentially requires that optimal solutions will not converge arbitrarily fast to $\bar x$, which is reasonable to expect for most physical systems especially.\footnote{Especially, if the Jacobian linearization of \eqref{eq:sys} at $\bar z\in \inte{Z}$ is stabilizable and the stage cost $\ell$ is quadratic in $u$, then such a local bound could be constructed from the solution to the algebraic Riccati equation.}
\begin{assumption}[Controllability or stabilizability] \label{ass:ctrb}
The Jacobian linearization of \eqref{eq:sys} at  $(\bar x, \bar u)$, $(A, B)$ $\doteq(f_x, f_u)$, is a) controllable, or b) stabilizable.
\end{assumption}

\begin{theorem}[\eqref{eq:strDI} $\Rightarrow$ asymptotic primal stability]\label{thm:asympStab}
Let OCP$_T(x_0)$ be strictly dissipative at $\bar z = (\bar x, \bar u)\in \inte{Z}$, and let Assumptions  \ref{ass:contW} and \ref{ass:ctrb}a hold. 
Suppose that $V_\infty$ and some storage function $S$ are $\mathcal{C}^1$ on  an open neighborhood of $\bar x$. Then, for all $ x_0 \in \mathbb{X}_0$, the point $\bar x$ is locally asymptotically stable for the solutions of \eqref{eq:sys}. If, additionally, Assumption \ref{ass:reach}  holds, then $\mbb{X}_0$ is in the region of attraction.
\end{theorem}

\begin{proof}
We consider the Lyapunov function candidate $W$ from \eqref{eq:W}. 
From Lemma \ref{lem:Vinf0} it follows that $W(\bar x) = 0$. 
Moreover, the strict dissipation inequality \eqref{eq:strDI} can be written as
\[
\int_0^\infty \alpha_\ell(\left\|z^\star(t, x)-\bar z\right\|)  \mathrm{d}t  \leq V_\infty( x) 
+ S(x) -  S(\bar x).
\]
As shown in \cite[Prop. 1]{Polushin05a}, controllability of the  Jacobian linearization of \eqref{eq:sys} at  $(\bar x, \bar u)$ implies on a neighborhood of $\bar x$
that
\begin{align*}
S(x) -  S(\bar x) & \leq \overline{\alpha}_S(\|x-\bar x\|),\quad\, \alpha_S \in \mcl{K}\\
 V_\infty(x)&\leq \overline{\alpha}_V(\|x-\bar x\|),\quad  \alpha_V \in \mcl{K}.
\end{align*}
 The usual derivative along the trajectories of \eqref{eq:sys} gives
\[
\dot W = \nabla W^\top f^\star(t, x) = (\nabla V_\infty  + \nabla S)^\top f^\star(t, x).
\]
Recall that the differential counterpart of \eqref{eq:strDI} reads
\[
 \nabla S^\top f^\star(t, x) \leq -\alpha_\ell(\|(x,u^\star(t, x_0)) -\bar z\|) + \ell(x,u^\star(t, x_0)).
\]
Hence, we have 
\[
\dot W \leq  -\alpha_\ell(\|(x,u^\star(\cdot)) -\bar z\|) + \ell(x,u^\star(\cdot)) + \nabla V_\infty^\top f^\star(t, x).
\]
Now, consider a neighborhood $\tilde{\mbb{X}}$ of $\bar x$ where  $\nabla V_\infty$ is $\mcl{C}^1$.
Recall that the adjoint variable $\lambda^\star$ corresponds to the gradient of the optimal value function, i.e. $\nabla V_\infty = \lambda^\star$ \eqref{eq:AdjointNablaV}.
Then on $\tilde{\mbb{X}}$ we have
\[
 \ell(x^\star,u^\star(\cdot)) + \nabla V_\infty^\top f^\star(t, x) = H(\lambda_0^\star, \lambda^\star, \mu^\star, x^\star, u^\star) \equiv 0,
\]
where the second equality follows from \eqref{eq:HamiltonZero}. Via $-\alpha_\ell(\|(x,u) -\bar z\|) \leq -\alpha_\ell(\|x -\bar x\|)$, we arrive at
\begin{align*}
\underline{\alpha}_W(\|x - \bar x\|) \leq W(x) \leq& \phantom{-}\overline{\alpha}_W(\|x-\bar x\|)\\
\dot W(x) \leq&  %\phantom{_S}
-\alpha_\ell(\|x -\bar x\|)
\end{align*}
where $\overline{\alpha}_W(\|x-\bar x\|) \doteq \overline{\alpha}_S(\|x-\bar x\|) + \overline{\alpha}_V(\|x-\bar x\|)$.
\end{proof}
One may wonder whether the  normalization of $W$ with $-S(\bar x)$ in \eqref{eq:W} can be avoided. The next lemma gives conditions which imply $S^a_{\alpha, \ell}(\bar x) = 0$.
\begin{lemma}[$S^a_{\ell}(\bar x) = 0$] \label{lem:availStor0}
Let OCP$_T(x_0)$ be strictly dissipative at $\bar z = (\bar x, \bar u)\in\mbb{Z} $  and let $\ell(\bar z) = 0$ and $\ell(z^\star(t)) \leq 0$ along optimal solutions. Then the available storage for non-strict dissipativity satisfies $S^a_{\ell}(\bar x)=0$.
Moreover, if \eqref{eq:sys} has an equilibrium at $\bar z= (\bar x, \bar z)$, then also the available storage for strict dissipativity satisfies $S^a_{\alpha, \ell}(\bar x) = 0$. 
\end{lemma}
\begin{proof}
For the case of non-strict dissipativity, it follows from \eqref{eq:availStorage} that for all $x \in \mbb{X}_0$ the equality $S^a_{\ell}(x) = -V_\infty(x)$ since $-\ell(x,u) \geq 0$, and the supremum in \eqref{eq:availStorage} is attained for $T \to \infty$. Hence from Lemma \ref{lem:Vinf0} we have $S^a_{\ell}(\bar x) = -V_\infty(\bar x) = 0$.
In case of strict dissipativity, observe that an equilibrium of \eqref{eq:sys} at $\bar z= (\bar x, \bar z)$ implies that the optimal solution of OCP$_\infty(\bar x)$ is unique and stationary. Hence the class $\mcl{K}_\infty$ function $\alpha_\ell$ in \eqref{eq:strAvailStorage} equates to $0$ almost everywhere. Thus   we have $S^a_{\alpha, \ell}(\bar x) = 0$. 
\end{proof}
Next, we combine stability of the optimally controlled  system with polynomial bounds on the  Lyapunov function to obtain exponential stability/reachability.

\begin{assumption}[Polynomial bounds on $W$]\label{ass:polynomBndW}~\\
The function $W:\Rnx\to\mbb{R}^+_0$ from \eqref{eq:W} is bounded from above by $\overline{\alpha}_W$ and from below by $\underline{\alpha}_W$ such that 
\begin{align*}
%\alpha_\ell(\|x -\bar x\|) & = c_\ell\|x -\bar x\|^{p_\ell}, \qquad~ p_\ell \geq 1, \,~c_\ell > 0 \\
\overline{\alpha}_W(\|x -\bar x\|) &= w_1\|x -\bar x\|^{w_2}, \quad \quad w_1> 0, ~ w_2\geq 1,  \\
\underline{\alpha}_W(\|x -\bar x\|) &= w_3\|x -\bar x\|^{w_4}, \quad \quad w_3> 0, ~ w_4\geq 1.
%\vspace*{-7mm}
\end{align*}
%i.e. $W$ is polynomially bounded.

\end{assumption}
\begin{proposition}[\eqref{eq:strDI} $\Rightarrow$ exp. reachability]  \label{prop:expReach}
For all $x_0\in \mbb{X}_0$,  let OCP$_T(x_0)$ be strictly dissipative at $\bar z = (\bar x, \bar u)\in\inte{Z} $ with polynomial strictness, i.e. \eqref{eq:strDI} holds with 
%\begin{subequations}\label{eq:alphaK_bnd}
\begin{equation}
\alpha_\ell(\|x -\bar x\|) = c_\ell\|x -\bar x\|^{p_\ell}, \qquad~ p_\ell \geq 1, \,~c_\ell > 0.
\end{equation}
Suppose that the conditions of Theorem \ref{thm:asympStab} hold, and let Assumption \ref{ass:polynomBndW} hold. 

Then, for all $x_0\in \tilde{\mbb{X}}_0$, there exists constants $C >0, \rho >0$ such that the optimal infinite-horizon state responses satisfy 
\[
\|x^\star(t,x_0, u^\star(\cdot, x_0)) - \bar x\| \leq Ce^{-\rho t},
\]
i.e. $\bar x$ is exponentially reachable, where $ \tilde{\mbb{X}}_0$ is any compact subset of $\mbb{X}_0$.
\end{proposition}
\begin{proof}
If the  functions $\alpha_\ell \in \mcl{K}_\infty$, $\overline{\alpha}_W \in \mcl{K}$, and $\underline{\alpha}_W \in \mcl{K}$ are polynomial, 
then Theorem \ref{thm:asympStab} implies local exponential stability of \eqref{eq:sys}.
It remains to show that, for all $x_0\in \tilde{\mbb{X}}_0$, a sufficiently small $\varepsilon$-neighborhood of $\bar x$ is reached in finite time $t(\varepsilon)$ and independently of $x_0$.

To this end, observe that, on any compact set $ \tilde{\mbb{X}}_0$, Assumption \ref{ass:polynomBndW} implies $W(x) \leq C_W$. Moreover, consider the set
$
\Theta_{\varepsilon, \infty}(\bar x) \doteq  \{t\in [0,\infty)\,|\, \|x^\star(t,x_0) - \bar x \| > \varepsilon\}$. The inequality \eqref{eq:strDI} yields
\[
\int_0^\infty \alpha_\ell(\|x^\star(t,x_0) - \bar x \|)\mathrm{d}t \leq W(x_0) \leq C_W.
\]
Similar to turnpike results in \citep{epfl:faulwasser15h}, splitting the integration domain $[0,\infty)$ into $\Theta_{\varepsilon, \infty}(\bar x)$ and $[0,\infty)\setminus \Theta_{\varepsilon, \infty}(\bar x)$ we obtain
$\textstyle
\nu[\Theta_{\varepsilon, \infty}(\bar x)] \leq \frac{C_W}{\alpha_\ell(\varepsilon)}$, where
$\nu[\cdot]$ is the Lebesgue measure on the real line. 
Note that this inequality implies that the time the optimal state response spends outside of $\mcl{B}_{\varepsilon}(\bar x)$ is bounded, independent of $T$ and $x_0\in \tilde{\mbb{X}}_0$, from above by $\frac{C_W}{\alpha_\ell(\varepsilon)}$.

There exists $\bar\varepsilon>0$ and $\varrho >0$ such that $\mcl{B}_{\bar \varepsilon}(\bar x) \subseteq \Omega(\varrho) \subset \Rnx$ with 
\begin{align*}
\Omega(\varrho)&\doteq \left\{x\,|\, W(x) \leq w_1 \varrho^{w_2}, \dot W(x) \leq -c_\ell\|x-\bar x\|^{p_\ell}\right\}
\end{align*}
i.e. some $\bar \varepsilon$ neighborhood of $\bar x$ will be contained in a  level set of $W(x)$ on which local exponential stability of \eqref{eq:sys} holds. Now pick any $\delta, \varepsilon \in (0,\bar\varepsilon)$ and consider $t(\varepsilon) > \frac{C_W}{\alpha_\ell(\delta)}$ with $\delta < \varepsilon$ such that $\mcl{B}_\delta(\bar x) \subseteq \Omega(\varepsilon)\subseteq \mcl{B}_\varepsilon(\bar x)$. 
If $\|x^\star(t(\varepsilon),x_0) - \bar x \| \leq \delta < \varepsilon < \bar\varepsilon$, then for all $t \geq t(\varepsilon)$ the exponential convergence bound holds. Suppose that  $\|x^\star(t(\varepsilon),x_0) - \bar x \| > \delta$, then due to $\textstyle t(\varepsilon) > \frac{C_W}{\alpha_\ell(\delta)}$, there exist $\tilde t \in 
[0,t(\varepsilon))$ with $\|x^\star(\tilde t,x_0) - \bar x \| \leq \delta < \varepsilon < \bar\varepsilon$. Due to
 local exponential stability of \eqref{eq:sys}, once $\|x^\star(\tilde t,x_0) - \bar x \| < \delta$ the solution of \eqref{eq:sys} will not leave $\Omega(\varepsilon) \subseteq \mcl{B}_\varepsilon(\bar x)$ for $t > \tilde t$. Hence $\|x^\star(t(\varepsilon),x_0) - \bar x \| < \varepsilon$ and thus the exponential convergence bound holds. %This concludes the proof.
 \end{proof}

\subsection{Converse dissipativity results}
\begin{proposition}[Exp. primal stability $\Rightarrow$ \eqref{eq:strDI}] \label{prop:pstabDI}
Let the solutions of \eqref{eq:sys} be locally exponentially stable on some neighborhood $\mcl{B}(\bar x)$, let $\ell$ be Lipschitz continuous on $\mbb{Z}$ with constant $L_\ell$, and let  the optimal infinite-horizon inputs satisfy 
\begin{equation} \label{eq:BndUopt}
\|u^\star(t, x_0) - \bar u\| \leq C_u e^{-\rho_u t} \quad \text{with } C_u >0, \rho_u>0.
\end{equation}
Then, for all $x_0 \in \mcl{B}(\bar x)$, 
OCP$_T(x_0)$ is strictly dissipative at $\bar z = (\bar x, \bar u)$ with polynomial $\alpha_\ell(\|z -\bar z\|)$.
\end{proposition}
\begin{proof}
Lipschitz continuity of $\ell$ combined with $\ell(\bar z)=0$ gives
\begin{align*}
\int_0^{\infty} \ell(z^\star(t, x_0)) -\ell(\bar z)\mathrm{d}t &\leq 
\int_0^{\infty}\left| \ell(z^\star(t, x_0)) -\ell(\bar z) \right|\mathrm{d}t \\ &\leq 
L_\ell\int_0^\infty \|z^\star(t, x_0) - \bar z\| \mathrm{d}t . 
%\\ \leq  L_\ell\cdot C <\infty
\end{align*}
Now exponential stability of \eqref{eq:sys} and the exponential bound on the optimal infinite-horizon controls imply
\[
\int_0^{\infty} \ell(z^\star(t, x_0)) -\ell(\bar z)\mathrm{d}t \\\leq  L_\ell\cdot C <\infty.
\]
Recall that boundedness of the available storage certifies dissipativity. Hence we use \eqref{eq:strAvailStorage} and
consider
$
\alpha_\ell(s)\doteq c\cdot s.$ %\left(\|x(t, x_0, u^\star(\cdot)) - \bar x\| + \|u^%\star(x(t)) - \bar u\|\right)$.
As $\ell(\bar z)=0$ it follows immediately that, for any finite $c>0$ and all $T\geq 0$, the solutions of \eqref{eq:sys} satisfy
$
\bar S^a_{\alpha,\ell}(x_0) < \infty
$
for all $x_0 \in \mcl{B}(\bar x)$.
\end{proof}
Theorem 1 of \cite{Willems71a} establishes the equiavalence between dissipativity and boundedness of $V_\infty$ for linear-quadratic problems. Next we show equivalence for the non-linear setting.  
\begin{theorem}[$V_\infty(x) >- \infty$ $\Leftrightarrow$ \eqref{eq:diss}]\label{thm:VinfBnd}~\\
%For all $ x_0 \in \mathbb{X}_0$,  let Assumption \ref{ass:reach} hold. 
The following two statements are equivalent:
\begin{enumerate}
\item[(i)] There exists a non-negative function $S:\mbb{X} \to \R_0^+$, satisfying \eqref{eq:diss}  on $\mbb{X}_0\subseteq \mbb{X}\subset \Rnx$.
\item[(ii)] For all $x \in \mbb{X}_0$ and all $T\in \mbb{R}^+\cup \infty$, it holds that $V_T(x)>-\infty$.
\end{enumerate}
\end{theorem}
\begin{proof}
(i) $\Rightarrow$ (ii), i.e. dissipativity $\Rightarrow$ $V_\infty(x) >- \infty$, is shown in Lemma \ref{lem:strgOpt}. (ii) $\Rightarrow$ (i) follows from
\begin{align*}
S^a_{\ell}(x) & =  \sup_T \int_{0}^{T} -\ell(z^\star(t, x))\,\mathrm{d}t = \sup_T ~-V_T(x)%\\
\end{align*}
As (ii) states $-V_T(x) < \infty$ for all $T \in \mbb{R}^+\cup \infty$, we conclude
 boundedness of the (non-strict) available storage $S^a_{\ell}(x)$ in \eqref{eq:availStorage}.
 \end{proof}

\subsection{Infinite-horizon transversality conditions}
\begin{assumption}[Unique Lagrange multipliers] \label{ass:LICQ}
The Lagrange multipliers  $\lambda \in \Rnx, \mu\in\R^{n_g}$ in the steady-state optimization problem \eqref{eq:SOP} are unique.
\end{assumption}
We remark that \cite{Wachsmuth13} has shown that for NLPs uniqueness of multipliers is equivalent to the well-known Linear-Independence Constraint Qualification (LICQ) of non-linear programming.\footnote{It is also interesting to note that whenever $\bar z \in \inte{Z}$, which implies $\bar \mu =0$, then controllability of the Jacobian linearization (Assumption \ref{ass:ctrb}a) implies LICQ. To show this, consider new coordinates for $x$ such that the Jacobian linearization is in controllability canonical  form. }
\begin{theorem}[\eqref{eq:strDI} $+$ LICQ $\Rightarrow$ adjoint attractivity] \label{thm:adjointAttractivity}
For all $x_0\in \mbb{X}_0$,  let OCP$_T(x_0)$ be strictly dissipative at $\bar z = (\bar x, \bar u)\in \mbb{Z} $ and let Assumptions \ref{ass:reach}, \ref{ass:ctrb}b, and \ref{ass:LICQ} hold. 
Then, for all $x_0 \in \mbb{X}_0$,  the infinite-horizon adjoint $\lambda^\star(\cdot, x_0)$ satisfies
\[\displaystyle 
\lim_{t\to\infty} \lambda^\star(t, x_0) = \bar \lambda. \vspace*{-5.0mm}
\]
%%\vspace{-3.0mm}
\end{theorem}
\begin{proof}
For the sake of contradiction, assume  that while the optimal primal variables converge $z^\star(t, x_0) \to \bar z$, the adjoint would not. Upon primal convergence, 
the adjoint dynamics \eqref{eq:NCO_adjoint} with ``output'' \eqref{eq:NCO_u} read
\begin{subequations}\label{eq:adjoint_dyn}
\begin{align} 
\dot \lambda^\star &= - f_x^\top \lambda^\star - \ell_x - g_x^\top \mu^\star \\
0 &= \phantom{-} f_u^\top \lambda^\star + \ell_u + g_u^\top \mu^\star. \label{eq:adjoint_dyn_output}
\end{align}
\end{subequations}
Since the optimal pair $z^\star(t) = \bar z =  (\bar x, \bar u)$ is at steady state, so is the multiplier $\mu = \bar \mu$. Hence, the adjoint $\lambda$ has to evolve in the subspace of $\Rnx$ spanned by \eqref{eq:adjoint_dyn_output}.
Observe that stabilizability of $(f_x, f_u) \doteq (A,B)$ implies detectability of $(-A^\top, B^\top)$. Hence, the adjoints converge
\[
\lim_{t\to\infty} \lambda^\star(t) = \tilde\lambda
\]
to some equilibrium $\tilde\lambda$.
Assumption \ref{ass:LICQ}, i.e. LICQ, implies that $ \tilde\lambda=\bar\lambda$ is the unique steady state solution to \eqref{eq:adjoint_dyn}, hence the assertion follows. %
\end{proof}
\begin{theorem}[Gradients of $V_\infty$ and $S$] \label{thm:gradV} ~\\
Let OCP$_T(x_0)$ be strictly dissipative at $\bar z = (\bar x, \bar u)\in\inte{Z} $, let $\ell(\bar z) = 0$, and let Assumptions \ref{ass:reach},  \ref{ass:ctrb}b, and \ref{ass:LICQ} hold. If $V_\infty$ is differentiable at $\bar x$, then 
\[
\nabla V_\infty(\bar x) = \bar\lambda= -\nabla S(\bar x),
\]
where $S$ is any differentiable storage function which certifies strict dissipativity of OCP$_T(x_0)$ with respect to $\bar z$. \vspace*{-2mm}
\end{theorem}
\begin{proof}
The left hand side assertion follows from the second statement in Lemma \ref{lem:Vinf0}, $z_\infty^\star(t,\bar x) \equiv (\bar x, \bar u)$ combined with Assumptions \ref{ass:ctrb}b and \ref{ass:LICQ}.
That is,  Lemma \ref{lem:Vinf0} implies that the stationary tuple $(\bar x, \bar u, \bar \lambda, 0)$ is the unique optimal solution to  OCP$_\infty(\bar x)$.  The right hand side has been shown in \cite[Theorem 4]{kit:zanon18a}. \vspace*{-2mm}
\end{proof}

\begin{example}[Dissipativity of Halkin's example]\label{ex:Halkin2}
The optimal state response in Halkin's example is
\[
x^\star(t) = e^{-t}x_0 +1-e^{-t}. 
\]
It is easily verified that 
\[
V_\infty(x_0) = x_0 - \lim_{t\to\infty}x^\star(t, x_0)= x_0-1.
\]
Hence, according to  Theorem \ref{thm:VinfBnd} Halkin's example is a dissipative OCP. 
Moreover, the differential counterpart to \eqref{eq:diss} reads
$
\nabla S\cdot (1-x)u \leq -(1-x)u - \ell(\bar z)$.
It is obvious that the optimal steady-state performance implies $ \ell(\bar z)=0$.
Hence $S(x) = -x +1$ is a storage function for all $x \in [0,1]$.
Indeed this is the available storage, since
\[
S^a_\ell(x_0) = \sup_T ~ x^\star(T,x_0) -x_0 = 1-x_0 = -V_\infty(x_0).\vspace*{-2mm}
\]
However, the corresponding steady-state minimizer $\bar z$ is not unique in this case as any $\bar x = 1, \bar u \in [0,1]$ as well as $\bar x \in [0,1], \bar u = 0$ achieves $\ell(\bar z) = 0$. This implies that Halkin's example can not satisfy any strict dissipation inequality \eqref{eq:strDI}, since at steady state
\eqref{eq:strDI} reads
$
\alpha_\ell(\|z-\bar z\|) \leq \ell(z) -\ell(\bar z)$,
i.e., strictness of  \eqref{eq:strDI} implies uniqueness of $\bar z$. \vspace*{-2mm}
\end{example}

At this point two questions are natural: Do the adjoints inherit the stability properties of the primal variables? Moreover, what can be said if LICQ in the steady-state optimization problem \eqref{eq:SOP} does not hold?
\begin{theorem}[\eqref{eq:strDI} $\Rightarrow$ exp. conv. of adjoints] \label{thm:adjointConv}~\\
Let the assumptions and conditions of Theorem \ref{thm:asympStab} hold, let OCP$_T(x_0)$ be strictly dissipative at $\bar z\in\inte{Z}$ with polynomial $\alpha_\ell$, and suppose polynomial bounds on $V_\infty$ and $S$  via Assumption \ref{ass:polynomBndW}. Furthermore, let $V_\infty$ have a locally Lipschitz continuous gradient on a compact subset $\tilde{\mbb{X}}_0 \subseteq \mbb{X}$ with 
$\bar x\in \mathrm{int}\,\tilde{\mbb{X}}_0$. 

Then, there exist $C_\lambda >0, \rho >0$ such that for all $x_0 \in \tilde{\mbb{X}}_0$,
\[
\|\lambda^\star(t, x_0) -\bar\lambda\| \leq C_\lambda e^{-\rho t}. %\vspace{-8mm}
\]
\TF{If additionally, $V_\infty$ has a locally bi-Lipschitz gradient\footnote{A function $f:\Rnx\to\R$ is locally bi-Lipschitz if there exists $L\in\R^+$ such that $\frac{1}{L}\|x-y\| \leq \|f(x) -f(y)\| \leq L \|x-y\|$ holds locally.}  then the optimal infinite-horizon adjoints $\lambda^\star(\cdot, x_0)$ are locally uniformly asymptotically stable at $\bar\lambda$.}
\end{theorem}
\begin{proof}
Recall that locally around $\bar x$, Theorem \ref{thm:asympStab} supposes that $V_\infty$ is of class $\mcl{C}^1$. 
Hence, we use that 
%\[
$
\lambda^\star(t, x_0) = \nabla V_\infty(x^\star(t, x_0))$ from \eqref{eq:AdjointNablaV}. 
% \]
It follows that
\[
\lambda^\star(t,x_0) -\bar\lambda = \nabla V_\infty(x^\star(t, x_0)) - \nabla V_\infty(\bar x).
\]
%where we have used that $\bar\lambda = \nabla_x V_\infty|_{\bar x}$.
Let $L_{\partial V}$ be a Lipschitz constant of $\nabla V_\infty$ on $\tilde{\mbb{X}}_0$, then
\begin{equation*}\label{eq:proofThm11}
\|\lambda^\star(t,x_0) -\bar\lambda\| \leq L_{\partial V} \|x^\star(t,x_0) - \bar x\| \leq L_{\partial V}C e^{-\rho t},
\end{equation*}
where we have used the bound from Proposition \ref{prop:expReach}. This proves the convergence statement.

\TF{Due to local exponential stability of $\bar x$, this bound can be rewritten locally as
$\|\lambda^\star(t,x_0) -\bar\lambda\|  \leq \tilde C \|x_0 - \bar x\| e^{-\rho t}$.
Bi-Lipschitzness of $\nabla V_\infty$ with constant $L_{\partial V}$ gives that on $ \|x_0 - \bar x\| \leq \delta$ it holds that $ \|x_0 - \bar x\| \leq L_{\partial V} \|\lambda_0 - \bar \lambda\|$.
Hence we obtain the local bound
\[
\|\lambda^\star(t,x_0) -\bar\lambda\|  \leq  \tilde C L_{\partial V}(\|\lambda_0 - \bar \lambda\|)  e^{-\rho t}.
\]
Invoking the definition of uniform asymptotic stability \citep[Def. 4.4]{Khalil02} concludes the proof.}
\end{proof}
It remains to analyze the case without LICQ. 
Let $(\tilde x, \tilde u, \tilde \lambda,\tilde\mu)$ be a  (not necessarily optimal) solution to the KKT conditions \eqref{eq:KKT} of \eqref{eq:SOP}, i.e. a KKT point. 
Consider the set 
\[
\Omega \doteq \left\{(\tilde \lambda,\tilde\mu) \,|\, (\tilde x, \tilde u, \tilde \lambda,\tilde\mu) \text{ solves } \eqref{eq:KKT}\right\},
\]
i.e., the set of all dual KKT solutions to \eqref{eq:SOP}. 
Let $\Pi_\lambda: (\lambda, \mu) \mapsto \lambda$ denote the projection on the adjoints $\lambda$.
\begin{proposition}[Transversality w/o LICQ] \label{prop:adjointTrans}~\\
For all $x_0\in \mbb{X}_0$,  let OCP$_T(x_0)$ be strictly dissipative at $\bar z = (\bar x, \bar u)\in \mbb{Z}$ and suppose that Assumptions \ref{ass:reach} and  \ref{ass:ctrb}b hold. 
Moreover, suppose that there exists an optimal infinite-horizon input $u^\star(\cdot, x_0)$ absolutely continuous on $[0, \infty)$, and that 
OCP$_T(x_0)$ is non-singular at $(\bar x, \bar u)$, i.e. $\det H_{uu}(\bar x, \bar u) \not = 0$.

Then, for all $x_0 \in \mbb{X}_0$,  the infinite-horizon adjoint $\lambda^\star(\cdot, x_0)$ satisfies
$\displaystyle
\lim_{t\to\infty} \lambda^\star(t, x_0) = \tilde\lambda\in \Pi_\lambda\left(\Omega\right)$.%\vspace{-8mm}
\end{proposition}
\begin{proof}
The proof is similar to the one of Theorem \ref{thm:adjointAttractivity}. Assume that the primal variables converge due to strict dissipativity (see Theorem \ref{thm:attractivity}), while the adjoints do not.
Then the adjoint dynamics \eqref{eq:NCO_adjoint} with ``output'' \eqref{eq:NCO_u} are given by \eqref{eq:adjoint_dyn}.
Due to convergence  $z^\star(t) \to\bar z =  (\bar x, \bar u)$ the multiplier function satisfies $\mu \to \tilde \mu$. 

Note that the regularity of OCP$_T(x_0)$---$\det H_{uu}(\bar x, \bar u) \not = 0$---implies via the implicit function theorem that locally around $\bar z = (\bar x, \bar u)\in \mbb{Z}$, $u^\star(t,x_0)$ is a continuous function of $\lambda^\star(t,x_0), \mu^\star(t,x_0)$ as long as no changes in the active set occur. 
For the primal variables to stay at steady state, the adjoints have to be at least partially at steady state.
%, i.e. they have to evolve in $\Pi_\lambda\left(\Lambda \right)$. 
More precisely, all adjoints observable through $f_u^\top \doteq B^\top$ are at steady state, while the remaining adjoint modes have to be asymptotically stable (due to detectability of $(-f_x^\top, f_u^\top) \doteq (-A^\top, B^\top)$, cf. \eqref{eq:adjoint_dyn}. Hence 
$
\lim_{t\to\infty} \lambda^\star(t) = \tilde\lambda, ~ \tilde\lambda \in \Pi_\lambda(\Omega)$. %\vspace{-11mm}
%\]
%with $\tilde\lambda \in \Pi_\lambda(\Omega)$.
\end{proof}
 
\subsection{Equivalence of OCP Dissipativity and Stability}
Finally, it remains to answer the question for equivalence of strict dissipativity and dual/adjoint stability. 
\begin{assumption}[$x-u$ regularity of OCP$_T(x_0)$] \label{ass:regularOCP}
The primal Hessian of the Hamiltonian $H$ of OCP$_T(x_0)$ defined in \eqref{eq:H}  is such that
\[
\det \nabla^2 H = \begin{vmatrix}
H_{xx} & H_{ux}\\ H_{xu} & H_{uu}
\end{vmatrix} \not = 0 
\] holds at $(1,\bar \lambda, \bar \mu, \bar x, \bar u)$.%
\end{assumption}
 We remark that the above condition is less strict than the one imposed by	 \cite{Trelat15a}.
 The next example shows that  this local regularity condition differs from the usual local non-singularity of OCPs (which would be $\det H_{uu} \not = 0$).
\begin{example}[$x-u$ regularity of OCPs]~\\
Consider any interior point $(x,u)\in\inte{Z}$ for Halkin's example. We obtain
$
\det \nabla^2 H = -( \bar\lambda-1)^2, 
%\begin{vmatrix}
%0 & -1 + \bar\lambda \\
%-1 + \bar\lambda & 0
%\end{vmatrix}
$
which is non-singular for $\bar \lambda \not = 1$. However, in Example \ref{ex:halkin} we have shown that $\lambda^\star(t) \equiv \lambda_0^\star = 1$. It is easy to see that in Halkin's example $\ell(x,u) = -f(x,u)$ implies the  choice $\bar \lambda = 1$ (provided the normalization $\lambda_0 = 1$).
Hence, for $\bar \lambda = 1$ Halkin's example  it is not $x-u$ regular. 
%as the gradient of As Halkin's example does not satisfy LICQ (Assumption \ref{ass:LICQ}) one could choose $\bar\lambda \not = 1$. 

The optimal fish harvest discussed in Example \ref{ex:fish} satisfies LICQ at $\bar z$. We obtain $\det \nabla^2 H =-\left(\frac{b}{\bar x}\right)^2$.
Observe that for this example $x-u$ regularity is satisfied, while the OCP as such is singular, i.e. $\det H_{uu} = 0$. %
\end{example}

\begin{proposition}[Exp. conv. of duals $\Rightarrow$ \eqref{eq:strDI}] \label{prop:dstabDI}
Consider OCP$_\infty(x_0)$ with problem data being at least locally $\mcl{C}^2$ in $x$ and $u$, and 
%let $\ell$ be Lipschitz on $\mbb{Z}$ and 
$\bar z \in \inte{Z}$. Let Assumption \ref{ass:LICQ} and \ref{ass:regularOCP} hold and suppose that
there exist constants $C_\lambda >0, \rho_\lambda >0, C_\mu >0, \rho_\mu >0$ such that for all $x_0$ from an open neighborhood $\mcl{B}(\bar x)$,
\begin{subequations} \label{eq:expDualStab}
\begin{align}
\|\lambda^\star(t, x_0) -\bar\lambda\| &\leq C_\lambda e^{-\rho_\lambda t}\\
\|\mu^\star(t, x_0) -\bar\mu\| &\leq C_\mu e^{-\rho_\mu t}. \label{eq:expDualStab_mu}
\end{align}
\end{subequations}
Then, for all $x_0\in \mcl{B}(\bar x)$, OCP$_T(x_0)$ is strictly dissipative  at $\bar z = (\bar x, \bar u)$ with polynomial $\alpha_\ell(\|z -\bar z\|)$.%
\end{proposition}
\begin{proof}
We rewrite the adjoint dynamics \eqref{eq:NCO_adjoint} with ``output'' \eqref{eq:NCO_u} as an implicit system of equations
\begin{align*}
F\left(\dot\lambda^\star, \lambda^\star, \mu^\star, x^\star, u^\star\right) \doteq  \begin{bmatrix}
\dot\lambda^\star + H_x \\
H_u
\end{bmatrix}=0. 
\end{align*}
Observe that $F(0, \bar\lambda, 0, \bar x, \bar u) = 0$, where $\bar \lambda$ is unique due to LICQ. Now the condition $\det \nabla^2 H  \not = 0$ implies that 
$F(0, \bar\lambda, 0, \bar x, \bar u) = 0$ locally admits an implicit function 
$
z^\star  = F^{-1}\left(\dot\lambda^\star, \lambda^\star, \mu^\star\right)$.
Hence we obtain
\[
\|z^\star - \bar z \| = \left\|F^{-1}\left(\dot\lambda^\star, \lambda^\star, \mu^\star\right) - F^{-1}\left(0, \bar \lambda, \bar \mu\right)\right\|.
\]
As the  problem data of OCP$_\infty(x_0)$ is $\mcl{C}^2$ in $x$ and $u$, we have $ F^{-1} \in \mcl{C}^1$ and thus it is locally Lipschitz. Therefore
\[
\left\|F^{-1}\left(\dot\lambda^\star, \lambda^\star, \mu^\star\right) - F^{-1}\left(0, \bar \lambda, \bar \mu\right)\right\| \leq L_F\left\|
\begin{matrix}
\dot \lambda^\star(t, x_0) \\
\lambda^\star(t, x_0) -\bar\lambda\\
\mu^\star(t, x_0) -\bar\mu\end{matrix}
\right\|
\]
can be simplified to yield
\[
\left\|F^{-1}\left(\dot\lambda^\star, \lambda^\star, \mu^\star\right) - F^{-1}\left(0, \bar \lambda, \bar \mu\right)\right\| \leq L_F 
\tilde C e^{-\tilde \rho t}
\]
and hence
$
\|z^\star(t, x_0) - \bar z \|\leq L_F 
\tilde C e^{-\tilde \rho t}$.
Applying Proposition \ref{prop:pstabDI} yields the assertion. 
\end{proof}
Naturally, the exponential decay of the constraint multiplier $\mu$ in \eqref{eq:expDualStab_mu} is difficult to check. However, whenever for all solutions originating from $x_0$ close to $\bar x$ the active set is empty---for varying $x_0$ and along the horizon---one has that $\mu^\star(t,x_0) \equiv 0 = \bar \mu$. Likewise, pure input constraints simplify things, as in this case we can drop \eqref{eq:NCO_u} from the optimality conditions.

Next, we establish a set of conditions under which exponential stability of \eqref{eq:sys}, exponential stability of the dual variables \eqref{eq:expDualStab} and strict dissipativity of OCP$_T(x_0)$ are equivalent. 
%To this end, 

\begin{theorem}[Local equivalence conditions]\label{thm:equivalence}~\\
Consider OCP$_T(x_0)$ with $T \in \mbb{R}^+\cup\infty$, let the problem data be at least locally $\mcl{C}^2$ in $x$ and $u$, let $\ell$ be  Lipschitz, and let
$\bar z \in \inte{Z}$. Suppose that $V_\infty$ and some storage function $S$ are $\mathcal{C}^2$ on  an open neighborhood $\mcl{B}(\bar x)$ and that  \eqref{eq:BndUopt} holds.  Furthermore, let Assumptions \ref{ass:reach}--\ref{ass:regularOCP} hold. 

Then, there exists an open neighborhood $\mcl{B}(\bar x)$ such that for all $x_0 \in \mcl{B}(\bar x)$  the following statements are equivalent:
\begin{enumerate}
\item[(i)] OCP$_T(x_0)$ is strictly dissipative with respect to $\bar z \in \inte{Z}$ and $\alpha_\ell$ polynomial. 
\item[(ii)] The optimal equilibrium $\bar x$ is exponentially stable for all infinite-horizon optimal solutions $x^\star(\cdot, x_0, u^\star(\cdot))$.
% are exponentially stable at $\bar x$.
\item[(iii)] The infinite-horizon optimal adjoints $\lambda_\infty^\star(\cdot, x_0)$ converge exponentially fast to the steady-state multiplier $\bar\lambda$. \vspace*{-2mm}
\end{enumerate}
\end{theorem}
\begin{proof}
(i) $\Rightarrow$ (ii) follows from Proposition \ref{prop:expReach} combined with Theorem \ref{thm:asympStab}. 
(i) $\Rightarrow$ (iii) follows from Theorem \ref{thm:adjointConv}, where $V_\infty$ and some storage function $S$ are locally $\mathcal{C}^2$, implies local Lipschitz continuity of $\nabla V_\infty(\bar x)$. 
%(ii) $\Rightarrow$ (i) is given by Proposition \ref{prop:pstabDI}.
(iii) $\Rightarrow$ (i) is shown in Proposition \ref{prop:dstabDI} using Assumption \ref{ass:regularOCP}. Note that due to $\bar z \in \inte{Z}$ we do not need \eqref{eq:expDualStab_mu}. 
%Moreover note that the proof of  Proposition \ref{prop:dstabDI}  also establishes (iii) $\Rightarrow$ (ii) in a weaker form.
Finally, (ii) $\Rightarrow$ (i) is given by Proposition \ref{prop:pstabDI} using  \eqref{eq:BndUopt} and (ii) $\Rightarrow$ (iii) can be shown as in the proof of Theorem \ref{thm:adjointConv}.\vspace*{-2mm}
\end{proof}

\section{Discussion ~\vspace*{-2mm}} \label{sec:discussion}
\begin{figure*}
\begin{center}
\includegraphics[width=0.95\textwidth]{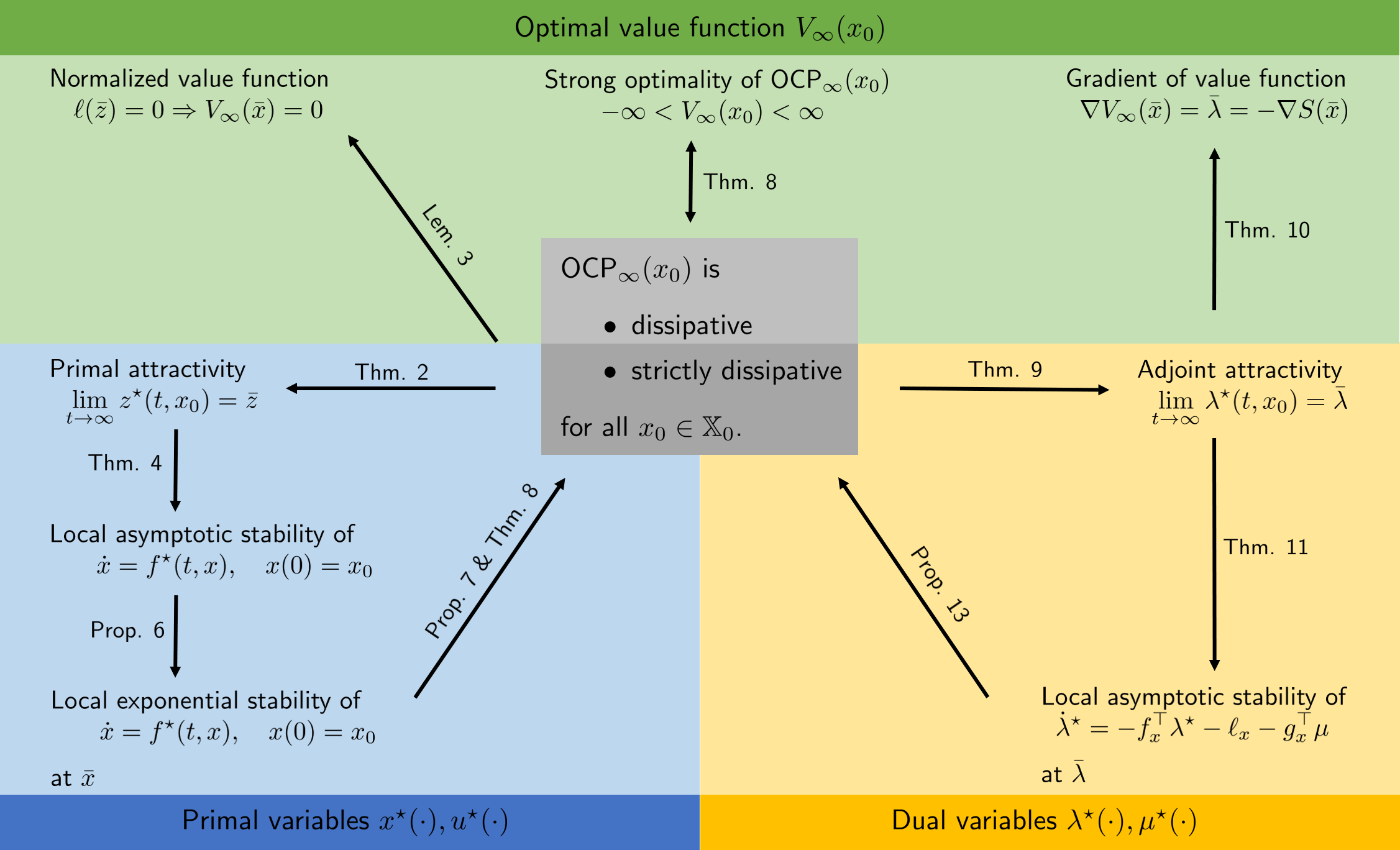}
\end{center}
\caption{Overview and relation of established results	. \label{fig:overview}}
\end{figure*}

\subsection{Overview of results~\vspace*{-3mm}}
Figure \ref{fig:overview} sketches the relations between the established results. In essence we have derived relations between strict dissipativity of OCP$_T(x_0)$, the corresponding primal variables $z^\star(\cdot) = (x^\star(\cdot), u^\star(\cdot))^\top$,  the dual variables $\lambda^\star(\cdot), \mu^\star(\cdot)$, and the value function $V_\infty(x_0)$.
Figure \ref{fig:overview} can also be viewed as a graphical illustration of the  proof of Theorem \ref{thm:equivalence}, cf. the cycles in the lower half.\vspace*{-2.0mm} 

%\subsubsection*
\textit{Dissipativity and the local geometry of $V$ and $S$.~}
First observe that the simple trick to offset the stage cost by $\ell(\bar z)$ allows to establish the  equivalence of (non-strict) dissipativity of OCP$_\infty(x)$ and boundedness of $V_T(x)$ for all $T \in \mbb{R}^+\cup\infty$. In a discrete-time context a similar offset idea was used by \cite{Gruene17b} to analyse the relation of time-varying turnpikes and dissipativity.\vspace*{-3mm}

Moving to  the top right corner of Figure \ref{fig:overview}, we remark that the local characterization of the gradient of $V_\infty$ shown in  Theorem \ref{thm:gradV} has already been hinted at in approximate fashion---i.e. $\nabla V_T(\bar x) \approx \bar \lambda$---in \cite[Remark 2]{kit:zanon18a}. Here we have strengthened this to the equivalence 
$
\nabla V_\infty(\bar x) = \bar \lambda = -\nabla S(\bar x).
$
Note that this relation holds for any  storage function $S$ that is locally differentiable at $\bar x$. We also remark that the right hand side relation has already been shown by \cite{Diehl11a}.
This very useful property is the key behind the construction of the Lyapunov function in \eqref{eq:W} as it allows to compensate for the non-vanishing gradients of both $V_\infty(x)$ and $S(x)$ at $\bar x$, while $S(\bar x)$ normalizes $W(x)$ to be $0$ at $\bar x$.\vspace*{-1mm}

%\subsubsection*
\textit{Transversality conditions of infinite-horizon OCPs.~}
Moving to the mid right hand side section of Figure \ref{fig:overview} we recall that 
Theorem \ref{thm:adjointAttractivity} and Theorem \ref{thm:adjointConv} establish attractivity and stability properties 
for the adjoint $\lambda$, i.e. 
$
\lim_{t\to\infty} \lambda^\star(t,x_0) = \bar\lambda. 
$
Recall that the steady-state adjoint equations correspond to the KKT conditions of the steady-state problem \eqref{eq:SOP}. Hence Theorem \ref{thm:adjointAttractivity} implies that
\begin{equation} \label{eq:adjointTrans}
\lim_{t\to\infty}  f_x^\top \lambda^\star(t) + \ell_x + g_x^\top \mu^\star(t) = 
f_x^\top \bar\lambda +\ell_x +g_x^\top\bar\mu = 0.
\end{equation}
Indeed, this equation can be seen as the infinite-horizon transversality condition in case the mixed input state constraints \eqref{eq:OCP_con} are active at $t =\infty$.
The main assumptions of Theorem \ref{thm:adjointAttractivity} fall in three categories: (i) strict dissipativity, (ii) stabilizability of the Jacobian linearization of \ref{eq:sys} at $(\bar x, \bar u)$, and (iii) LICQ (Assumption \ref{ass:LICQ}) in the steady-state problem \eqref{eq:SOP}. The latter two properties are needed to analyze the dynamics of the adjoints. 
Note that Proposition \ref{prop:adjointTrans}, not depicted in Figure \ref{fig:overview}, relaxes the LICQ requirement. It appears too difficult to further relax the stabilizability requirement.\vspace*{-3mm}

The value of Theorem \ref{thm:adjointAttractivity} lies in leveraging strict dissipativity assumptions to answer the open problem of adjoint transversality conditions for infinite-horizon OCPs, which dates back to the seminal observations of \cite{Halkin74a}. Therein, Halkin observed that the usual finite-horizon transversality condition---given by $\lambda^\star(T,x_0) = 0$ in  the absence of a Mayer term---does not carry over to the infinite-horizon case. 
Theorem \ref{thm:adjointAttractivity} closes this gap by utilizing strict dissipativity of OCPs to derive an asymptotic adjoint transversality condition for infinite-horizon optimal control problems via the steady state adjoint $\bar\lambda$. It is worth noting that from the dissipativity and turnpike point of view \citep{epfl:faulwasser15h,Trelat15a}---especially considering the concept of exponential turnpike properties---this adjoint attractivity is quite natural. 

%\subsubsection*
\textit{Stability of the optimality system.~}
Next, we focus on the lower half of Figure \ref{fig:overview} which is concerned with the stability properties of primal and dual variables. We remark that global asymptotic stability of optimal steady state for the Hamiltonian optimality system \eqref{eq:NCO} (modulo removing the algebraic constraint $0=H_u$ by substitution) has been studied by  \cite{Brock76a}, see also \citep[Chapter 4.3 and Theorem 4.4]{Carlson91}. Therein stability for maximization problems is established using $-\lambda^\top x$ as a Lyapunov function and by imposing definiteness assumptions on the Hessian of the maximized Hamiltonian. In contrast our analysis does not require such	 assumptions. 
Moreover, note that $W$ from \eqref{eq:W} is similar in construction to Lyapunov functions for practical stability used for economic MPC in \citep{epfl:faulwasser15g,Gruene17a}.

\subsection{Links to existing results \vspace*{-2mm}}
%\subsubsection*
\textit{Link to turnpike and dissipativity results.~}
Beyond the illustration in Figure \ref{fig:overview} our results complete a picture of dissipativity implications for OCPs which has been triggered by investigations of economic MPC schemes \citep{Angeli12a,Mueller14a,Gruene16a} in discrete-time and  \citep{epfl:faulwasser15h} in continuous-time. Turnpike properties have originally been observed in OCPs arising in economics; they refer to a similarity property of parametric OCPs, where for varying initial conditions and varying horizons the optimal solutions spend most of their time close to the optimal steady-state (a.k.a.  \textit{the turnpike}), see \citep{Dorfman58,Mckenzie76,Carlson91,tudo:faulwasser21b}.
Importantly, our results complement the analysis of (near) equivalence of turnpike and dissipativity properties for OCPs \citep{Gruene16a,epfl:faulwasser15h} to the aspect of infinite-horizon stability. 
We remark that in contrast to the practical stability result by \cite{Gruene17c}, we show  local asymptotic stability.%\vspace*{-3mm}

%\subsubsection*
\textit{Dissipativity and economic MPC.~}
In light of the results presented above, the stability analysis for economic MPC schemes conducted in  \citep{kit:faulwasser18e_2,kit:zanon18a}, which is  based on linear terminal penalties (i.e. Mayer terms) of the form $V_f(x) = \bar\lambda^\top x$, can be understood quite directly. As we have shown $V_\infty(\bar x) = 0$ (Lemma \ref{lem:Vinf0}) and $\nabla V_\infty(\bar x) = \bar\lambda$, hence $V_f(x) = \bar\lambda^\top x$ can be interpreted as a first-order approximation of the cost-to-go.  We remark that this simple trick closes the gap between practical stability \citep{Gruene13a,epfl:faulwasser15g} and asymptotic stability  \citep{kit:faulwasser18e_2,kit:zanon18a} in dissipativity approaches to economic MPC.
A geometric interpretation as gradient correction of the stage cost $\ell(x,u)$ (i.e. the Lagrange term) has been given by \cite{Zanon16a}.%\vspace*{-3mm}

%\subsubsection*
\textit{Dissipation inequalities and the HJBE.~}
Instead of the PMP one could as well employ the Hamilton-Jacobi-Bellmann Equation (HJBE) to solve the OCP \eqref{eq:OCP} at hand. We next sketch the relation between the dissipation inequalities \eqref{eq:diss} and \eqref{eq:strDI} and the HJBE. For simplicity suppose that the mixed input-state constraints \eqref{eq:OCP_con} reduce to pure input constraints defined via the set $\mbb{U}\subseteq \Rnu$. Under suitable differentiability assumptions on  $V_T(t, x)$---which we now write with two arguments  $t$ indicating initial time of the horizon $[t, T]$ and $x$ the initial condition---the HJBE is given by 
\begin{equation} \label{eq:HJBE}
 0= -\nabla_t V_T(t,x) - \inf_{u \in \mbb{U}}\left( \ell(x,u) + \nabla_x V_T^\top f(x,u)\right),  \tag{HJBE}
\end{equation}
where due to the absence of a Mayer term we have $V_T(T, x) = 0$.
In integral form we may write
\[
V_T(0, x_0) - V_T(T, x^\star(T)) = \int_0^T \ell(x^\star(t), u^\star(t)) \mathrm{d}t.
\]
Comparison with \eqref{eq:diss} shows that
\[
S(x^\star(T)) - S(x_0) \leq V_T(0, x_0) - V_T(T, x^\star(T)). 
\]
Recall the boundary condition $V_T(T, x) = 0$, hence we see that any storage function $S$ defines a lower bound on the optimal value function $V_T$. \vspace*{-3mm}

Suppose on some domain $\mcl{X}\subseteq\Rnx$ there exists a differentiable storage function $S$. Then the differential inequality from above can be written as
\begin{equation} \label{eq:subHJBEI}
  -\inf_{u \in \mbb{U}}\left( \ell(x,u) + \nabla_x S^\top f(x,u)\right) \leq 0.  
\end{equation}
In other words, any differentiable storage function defines a subsolution of \eqref{eq:HJBE}. Moreover, for all $x$, let the superdifferential of $S$ be denoted by $D^+ S(x)$ and \eqref{eq:subHJBEI} holds in the following sense
\begin{equation} \label{eq:subHJBEII}
 -\inf_{u \in \mbb{U}}\left( \ell(x,u) + \xi^\top f(x,u)\right) \leq 0,\quad \forall \xi \in D^+ S(x),
\end{equation}
then the storage function $S$ constitutes a viscosity subsolution of \eqref{eq:HJBE}, see \citep{Liberzon12}. Likewise, any bounded infinite-horizon viscosity subsolution of the HJBE will also constitute a storage function. 
Given the impact of viscosity solutions of the HJBE on optimal control theory, see e.g. \citep{Crandall84a,Bardi08a}, it is fair to ask for further links between storage functions and viscosity solutions. Moreover, recalling that controllability plays a pivotal role in establishing existence of continuous storage functions \citep{Polushin05a},  the link between viscosity solutions and storage functions might provide a road towards characterization of further regularity properties of the latter. %\vspace*{-3mm}

%\subsubsection*
\textit{Inverse optimality, feedback, and dissipativity.~}
The close interplay between dissipativity, stability, and optimal feedback design has been observed already by \cite{Moylan73a,Hill76a}, while \citep{Willems71a} appears to be the very first work in this direction.
Specifically, Willems shows that dissipativity is pivotal in analysing linear-quadratic infinite-horizon optimal control via characterization of the algebraic Riccati equation and related LMIs. 
Moreover, \cite{Moylan73a} show that under certain smoothness assumptions  passive output feedback for input affine systems is optimal with respect to a specific objective functional with essentially quadratic structure, while \cite{Freeman96a} discuss the inverse optimality problem of a given feedback. Recently, there have been extensions towards input quadratic systems \citep{Sassano19a}. These approaches have in common that they rely heavily on the existence of an appropriate differentiable solution to some associated HJBE. \vspace*{-3mm}

Our results differ as they do not provide analytic optimal feedbacks. However, they are similar in the sense that we discuss the stability \eqref{eq:sys} under optimal infinite-horizon controls. Moreover, our approach also includes constraints. Essentially, our results can be understood as generalization of \citep{Willems71a} to non-linear systems. Actually, we do not leverage the algebraic Riccati equation but include the adjoints in the analysis, which is only touched upon by Willems. This way, we also do not require the presence of a term $u^\top R u$ with $\det R \not = 0 $ in the objective, which in turn is pivotal in \citep{Willems71a}. Indeed our analysis also allows for non-quadratic stage costs $\ell$. Observe that Theorem \ref{thm:VinfBnd} is a generalization of Theorem 1 in \citep{Willems71a}. Finally, Theorem \ref{thm:asympStab}  presents a fairly general construction of a Lyapunov function for the non-linear infinite-horizon optimally controlled system.

\section{Conclusions} \label{sec:conclusion}

This paper has studied stability and dissipativity properties of infinite-horizon continuous-time optimal control problems with respect to primal and dual variables, i.e. with respect to inputs, states and adjoints. We have shown that strict dissipativity implies local exponential stability of infinite-horizon optimal solutions. We also derived converse statements, i.e. conditions under which stability of optimal solutions implies dissipativity.  Hence, the present paper can be considered an extension of the classical work by  \cite{Willems71a}. \vspace*{-2mm}

With respect to the adjoint variables the present paper addresses the issue of adjoint transversality conditions in infinite-horizon OCPs, which had been raised by \cite{Halkin74a}. Specifically, we have proven that strict dissipativity implies a natural adjoint characterization via the steady-state Lagrange multiplier.  We also established a formal equivalence between the gradients of the infinite-horizon optimal value function and any differentiable storage function at the optimal steady state, which is again characterized by the  steady-state Lagrange multiplier. 
Finally, this paper has put its results in perspective to recent developments on turnpike theory, on economic MPC, and on viscosity solutions of Hamilton-Jacobi-Bellmann Equations. \vspace*{-2mm}%, which warrant further investigations.  

\begin{ack} \vspace*{-2mm}                              % Place acknowledgements
The authors acknowledge the very helpful comments of the anonymous reviewers. 
TF acknowledges financial support by the Daimler and Benz Foundation. \vspace*{-2mm}
\end{ack}

\renewcommand*{\bibfont}{\scriptsize}
\bibliographystyle{plainnat}        % Include this if you use bibtex 
\bibliography{literature_latin1}           % and a bib file to produce the 
                                 % bibliography (preferred). The
                                 % correct style is generated by
                                 % Elsevier at the time of printing.

\end{document}